\pdfoutput=1
\documentclass{llncs}
\usepackage[margin=1in]{geometry} 
\usepackage{amssymb}
\usepackage{amsmath}
\usepackage{float}
\usepackage[newenum]{paralist}
\usepackage[usenames,dvipsnames]{xcolor}
\usepackage[noend]{algorithmic}
\usepackage{algorithm}
\usepackage[T1]{fontenc}
\usepackage{verbatim}
\usepackage[bottom]{footmisc}
\usepackage{calligra}
\usepackage{setspace} 
\usepackage{boxedminipage}
\usepackage{multirow}
\usepackage[titletoc,toc,title]{appendix}
\usepackage{hyperref}
\usepackage{graphicx}

\renewcommand{\lll}{\log\log\log n}

\algsetup{linenodelimiter=.}
\algsetup{indent=2em}

\hypersetup{colorlinks,linkcolor=blue,filecolor=blue,citecolor=blue,
urlcolor=blue,pdfstartview=FitH}

\newcommand{\cm}{\mathcal{CONGEST}}
\newcommand{\Ghat}{\hat{G}} 
\newcommand{\ehat}{\hat{E}}
\newcommand{\ct}{\mathcal{T}}
\newcommand{\cthat}{\hat{\mathcal{T}}}
\newcommand{\lra}{Lenzen's routing protocol}
\newcommand{\aspect}{the growth-bounded property}

\DeclareMathOperator{\poly}{poly}
\DeclareMathOperator{\degree}{degree}
\DeclareMathOperator*{\E}{\mathbf{E}}

\DeclareMathAlphabet{\mathpzc}{OT1}{pzc}{m}{it}

\newcommand{\ehh}{\mathpzc{h}}

\newcommand*{\union}{\mathop{\cup}}

\begin{document}

\title{Near-Constant-Time Distributed Algorithms on a Congested Clique\thanks{This work is supported in part by National Science Foundation grant CCF-1318166.}}
\author{James W. Hegeman \and Sriram V. Pemmaraju \and Vivek B. Sardeshmukh}
\institute{Department of Computer Science, The University of Iowa, Iowa City, IA 52242 
\email{\{james-hegeman, sriram-pemmaraju, vivek-sardeshmukh\}@uiowa.edu}}

\maketitle

\begin{abstract}
This paper presents constant-time and near-constant-time distributed algorithms for a variety of
problems in the congested clique model.
We show how to compute a 2-ruling set in $O(\log \log \log n)$ rounds with high probability and using this, 
we obtain a constant-approximation
to metric facility location, also in $O(\log \log \log n)$ rounds with high probability. 
In addition, assuming
an input metric space of constant doubling dimension, we obtain constant-round algorithms to compute
constant-factor approximations to the minimum spanning tree and the metric facility location
problems. These results significantly improve on the running time of the fastest known algorithms for
these problems in the congested clique setting. 
\end{abstract}

\section{Introduction}

The $\mathcal{CONGEST}$ model is a synchronous, message-passing model of distributed
computation in which the amount of information that a node can transmit along an incident communication link
in one round is restricted to $O(\log n)$ bits, where $n$ is the size of the network \cite{peleg2000distributed}.
As the name suggests, the $\mathcal{CONGEST}$ model focuses on \textit{congestion} as an obstacle to distributed computation.
In this paper, we focus on the design of distributed algorithms in the $\mathcal{CONGEST}$ model
on a \textit{clique} communication network; we call this the \textit{congested clique} model.
In the congested clique model, all information is nearby, i.e., at most one hop away,
and so any difficulty in solving a problem is due to congestion alone.

Let $H = (V, E_H)$ denote the underlying clique communication network.
In general, the input to the problems we consider consists of a $|V| \times |V|$ matrix $M$ of edge-attributes and a length-$|V|$ vector of node attributes.
$M$ represents edge weights (or distances, or costs) and it is initially distributed among the nodes in $V$ in such a way that each node $v \in V$ knows the corresponding
row and column of $M$.
In one typical example, $M$ could simply be the adjacency matrix of a spanning subgraph $G = (V, E)$ of 
$H$; in this setting, each node $v \in V$ initially knows all the edges of $G$ incident on it.
A number of classical problems in
distributed computing, e.g., maximal independent set (MIS), vertex coloring, edge coloring, maximal matching, shortest paths, etc., are well-defined in this setting. However, the difficulty of proving lower
bounds in the congested clique model \cite{DruckerKuhnOshmanPODC2014} means that it is not clear how quickly one should be 
able to solve any of these problems in this model. Note that the input $G$ can be quite dense (e.g., have 
$\Theta(n^2)$ edges)
and therefore any reasonably fast algorithm for the problem will have to be ``truly'' distributed in the sense
that it cannot simply rely on shipping off the problem description to a single node for local computation.
In this setting, the algorithm of Berns et al.~\cite{berns2012arxiv,berns2012facloc} that computes a \textit{2-ruling set} of $G$ in expected-$O(\log \log n)$ rounds is worth mentioning. 
(A \textit{t-ruling set} is defined to be an independent set $I \subseteq V$ such that every
node in $V$ is at most $t$ hops in $G$ from some node in $I$.) 
In another important class of problems that
we study, the input matrix $M$ represents a metric space $(V, d)$; thus each node $v \in V$ 
initially has knowledge of distances $d(v, w)$ for all $w \in V$. 
Nodes then need to collaborate to solve a problem such as \textit{minimum spanning tree} (MST) or 
\textit{metric facility location} (MFL) that are defined on the input metric space.
In this setting, the deterministic MST algorithm of Lotker et al.~\cite{lotker2006distributed} running in $O(\log\log n )$ 
rounds is worth mentioning.

Thus far the congested clique model has mainly served the theoretical purpose
of helping us understand the role of congestion as an obstacle to
distributed computation.
However, recent papers \cite{KlauckArxiv2013,HegemanPemmarajuSIROCCO2014}
have made connections between congested clique algorithms and
algorithms in popular systems of parallel computing such as MapReduce \cite{DeanGhemavat} and graph processing systems 
such as Pregel \cite{MalewiczPregelSIGMOD2010}, thus providing a practical motivation
for the development of fast algorithms on the congested clique.
Specifically, in \cite{HegemanPemmarajuSIROCCO2014}, it is shown that congested clique algorithms 
with fairly liberal resource constraints can be efficiently simulated 
in a MapReduce model of computation \cite{KarloffSuriVassilvitskii}.

\subsection{Main Results}
\label{section:mainResults}

In this paper we present several constant-time or near-constant-time algorithms for fundamental problems in the congested clique setting.

\begin{itemize}

\item First, we present an algorithm that computes a 2-ruling set of $G$
	in $O(\log \log \log n)$ rounds with high probability (in short, \textit{w.h.p.}, referring to probability at least $1 - 1/n^{c}$ for a constant $c \ge 1$), 
significantly improving the running time of the 2-ruling set
algorithm of Berns et al.~\cite{berns2012arxiv,berns2012facloc}.

\item Via a reduction presented in Berns et al.~\cite{berns2012arxiv,berns2012facloc}, this implies an $O(\log \log \log n)$-round algorithm
for computing an $O(1)$-approximation for MFL w.h.p. Again, this
significantly improves on the running time of the fastest known algorithm for this problem.
\end{itemize}

Distributed algorithms that run in $O(\log \log n)$ rounds are typically analyzed by showing a doubly-exponential rate of progress; such progress, for example, is achieved if the number of nodes that have ``successfully finished'' grows by squaring after each iteration. 
The congested clique algorithms for MST due to Lotker et al.~\cite{lotker2006distributed} and the above-mentioned MFL algorithm due to Berns et al.~\cite{berns2012arxiv,berns2012facloc} are both examples of such phenomena. Our algorithm with triply-logarithmic running time, involves new techniques that seem applicable to congested clique algorithms in general. Our result raises the distinct possibility that other problems, e.g., MST, can also be solved in $O(\log \log \log n)$ rounds on a congested clique.
In fact, our next set of results represents progress in this direction.

\begin{itemize}

\item We show how to solve the MIS problem on a congested clique in \textit{constant} rounds on an input graph $G_r$ 
induced by the metric space $(V, d)$ in which every pair of nodes at distance at most
$r$ (for any $r \ge 0$) are connected by an edge. This result has two implications.

\item First, given a metric space $(V, d)$ of constant doubling dimension, we show that a constant-approximation to
the MST problem on this metric space can be obtained in \textit{constant} rounds on a congested clique
setting.

\item An additional implication of the aforementioned MIS result is that it leads to a \textit{constant}-round constant-approximation to MFL in metric spaces of constant doubling dimension on a congested clique.
\end{itemize}

In order to achieve our results, we use a variety of techniques that balance bandwidth constraints with
the need to make rapid progress. We believe that our techniques will have independent utility in any
distributed setting in which congestion is a bottleneck.

\subsection{Technical Preliminaries}
\label{subsection:techPrelim}

\paragraph{Congested Clique Model.} The underlying communication network is a clique $H = (V, E_H)$ of size $n = |V|$.
Computation proceeds in synchronous rounds and in each round a node (i) receives all messages sent to it in the previous round, 
(ii) performs unlimited local computation, and then (iii) sends a, possibly 
different, message of size $O(\log n)$ to each of the 
other nodes in the network.
We assume that nodes have distinct IDs that can each be represented in $O(\log n)$ bits.

\paragraph{MST and MFL problems.} We assume that the input to the MST problem is a
metric space $(V, d)$. Initially, each node $v \in V$ knows
distances $d(v, w)$ to all nodes $w \in V$. When the algorithm ends, all nodes in $V$ 
are required to know a spanning tree $T$ of $V$ of minimum weight.
(Note that here we take $d(u, v)$ to be the ``weight'' of edge $\{u, v\}$.)
The input to MFL consists of a metric space $(V, d)$ 
along with \textit{facility opening costs} $f_v$ associated with
each node $v \in V$.
The goal is to find a subset $F \subseteq V$ of nodes to
\textit{open} as facilities so as to minimize the facility opening costs plus connection
costs, i.e., $\sum_{v \in F} f_v + \sum_{u \in V} D(u, F)$,
where $D(u, F) := \min_{v \in F} d(u, v)$ is the \textit{connection
cost} of node $u$.
Initially, each node $v \in V$ knows facility opening cost $f_v$ and distances $d(v, w)$ for
all $w \in V$.
Facility location is a well-studied problem in operations research
\cite{Balinski66,CNWBook,EHK77} that arises in
contexts such as locating hospitals in a city or locating distribution centers
in a region.
More recently, the facility location problem has been used as an abstraction for
the problem of locating resources in a wireless network
\cite{FrankBook,PanditPemmarajuICDCN09} and motivated by this application several distributed approximation algorithms for this problem have been
designed \cite{MoscibrodaFLPODC05,GehweilerSPAA2006,HegemanPemmarajuDISC2013}.

\paragraph{$t$-ruling set problem.}
A \textit{$t$-ruling set} of a graph $G = (V, E)$ is an independent set $I \subseteq V$ such that every vertex in $G$ is at most
$t$ hops from some vertex in $I$.
A $t$-ruling set, for constant $t$, is a natural generalization of an MIS and can stand as a proxy for an MIS in many instances.
The input to the $t$-ruling set problem on a congested clique $H = (V, E_H)$ is a spanning subgraph 
$G = (V, E)$ of the underlying communication network $H$.
Each node $v \in V$ is initially aware of all its neighbors in $G$.
At the end of the $t$-ruling set algorithm, every node is required to know the identities of 
all nodes in the computed $t$-ruling set.

\paragraph{Metric spaces, doubling dimension, and growth-bounded graphs.}
If $M = (V, d)$ is a metric space then we use $B_M(v, r)$ to denote the set of points $w \in V$ such that
$d(v, w) \le r$.
We call $B_M(v, r)$ the \textit{ball of radius $r$ centered at $v$}.
A metric space $M = (V, d)$ has \textit{doubling dimension} $\rho$ if for any $v \in V$ and $r \ge 0$,
$B_M(v, r)$ is contained in the union of at most $2^\rho$ balls $B_M(u, r/2)$, $u \in V$. 
In this paper, we work with metric spaces with constant doubling dimension, i.e., $\rho = O(1)$.
Note that constant-dimensional Euclidean metric spaces are natural examples of metric spaces with constant doubling dimension.
In distributed computing literature, metric spaces of constant doubling dimension have
been investigated in the context of wireless networks \cite{damian2006Spanner,KuhnMoscibrodaWattenhoferPODC2005}.
For a graph $G = (V, E)$ and a node $v \in V$, let $B_G(v, r)$ denote the set of all vertices $u \in V$ that are at most $r$ hops from $v$. 
A graph $G = (V, E)$ is said to have \textit{bounded growth} (or said to be \textit{growth-bounded}) if the size of any independent set in any ball
$B_G(v, r)$, $v \in V$, $r \ge 0$, is bounded by $O(r^c)$ for some constant $c$.
For any metric space $(V, d)$ and $r \ge 0$, the graph $G_{r} = (V, E_r)$, where 
$E_r = \{\{u, v\} \in d(u, v) \le r\}$
is called a \textit{distance-threshold graph}.
It is easy to see that if $(V,d)$ has constant doubling dimension then a distance-threshold graph
$G_{r}$, for any $r \ge 0$, is growth-bounded; this fact will play an important role in our algorithms.
For a given metric space $(V, d)$ the \emph{aspect ratio} $\lambda(Y)$ of a subset of points $Y \subseteq V$ is the ratio of maximum of pair-wise distance between points in $Y$ to the minimum of pair-wise distance between points in $Y$, i.e. 
 $\lambda(Y) = {\max\{d(u,v) \mid u, v\in Y\}}/{\min\{d(u,v)\mid u, v \in Y\}}$.
The following fact is easy to prove by applying the definition of doubling dimension: 
if $(V, d)$ is a metric with doubling dimension $\rho$ and $Y\subseteq V$ is a subset of points, then $|Y| \leq 2^{\rho\cdot\lceil \log_2 \lambda(Y) \rceil}$ where $\lambda(Y)$ is the 
aspect ratio of $Y$.
We refer to this property as the \textit{growth-bounded property} of the metric space $(V, d)$.
Distance-threshold graphs and more generally, growth-bounded graphs have 
attracted attention in the 
distributed computing community as flexible models of wireless networks \cite{KuhnMoscibrodaWattenhoferPODC2005}.
Schneider and Wattenhofer \cite{schneider2008logstar} present a deterministic algorithm,
running in $O(\log^* n)$ rounds, for computing an MIS on a growth-bounded graph.

\paragraph{Lenzen's routing protocol.}
A key algorithmic tool that allows us to design constant- and near-constant-time round 
algorithms is a recent deterministic routing protocol by Lenzen 
\cite{lenzen2013routing} that disseminates a large volume of information 
on a congested clique in constant rounds.
The specific routing problem, called an \textit{Information Distribution Task},
solved by Lenzen's protocol is the following.
Each node $i \in V$ is given a set of $n' \le n$ messages, each of size $O(\log n)$, $\{m_i^1, m_i^2, \ldots, m_i^{n'}\}$,
with destinations $d(m_i^j) \in V$, $j \in [n']$.
Messages are globally lexicographically ordered by their source $i$, destination $d(m_i^j)$, and $j$.
Each node is also the destination of at most $n$ messages.
Lenzen's routing protocol solves the Information Distribution Task in $O(1)$ rounds.

\paragraph{General Notation.} 
For a subset $S \subseteq V$, $G[S]$ denotes \textit{induced} subgraph of $G$ by set $S$; 
thus $G[S] = (S, E')$ where $E' = \{\{u, v\} \mid u, v \in S \mbox{ and } \{u, v\} \in E\}$.
In the context of our MST algorithm we will interpret metric distances $d(u,v)$ as as edge weights;
we will use $wt(u, v)$ and $d(u, v)$ interchangeably.
Given an edge-weighted graph $G = (V, E)$ and an edge set $E' \subseteq E$, 
we denote the sum of all edge-weights in $E'$ as $wt(E')$. 
We use $\Delta$ to denote the maximum degree of a graph; sometimes, to avoid 
ambiguity we use $\Delta(G)$ to denote maximum degree of graph $G$.
All logarithms are assumed to have base 2 unless otherwise specified.
We say an event occurs \textit{with high probability} (w.h.p.), if the probability of that event is at least $(1-1/n^c)$ for a constant $c \ge 1$.

\section{2-Ruling Sets in \texorpdfstring{$O(\log \log \log n)$}{O(log log log n)} Rounds}
\label{sec:ruling}
In this section, we show how nodes in $V$ can use the underlying clique
communication network $H$ to compute, in $O(\log \log \log n)$ rounds w.h.p,
a $2$-ruling set of an arbitrary spanning subgraph $G$ of $H$. 
At a high level, our $2$-ruling set algorithm can be viewed as having four steps. 
In the first step, the graph is decomposed into $O(\log \log n)$ degree-based classes and at the end of this step every node knows the class it belongs to. 
In the next subsection, we describe this \textit{degree-decomposition step} and show that it runs deterministically in $O(\log \log \log n)$ rounds.
In the second step, each vertex $v$ of the given graph $G$ joins a set $S$ independently with probability $p_v$,
where $p_v$ depends on $v$'s class as defined in the degree-decomposition step.  
This \textit{vertex-selection step} yields a set $S$ of nodes that will be shown to have two
properties: (i) w.h.p.~the number of edges in the induced subgraph $G[S]$ is
$O(n \cdot \poly(\log n))$; and (ii) w.h.p~, every vertex
in $G$ is either in $S$ or has a neighbor in $S$. 
Given the degree-decomposition, the vertex-selection step is elementary and requires no communication.
In the third step, we work with $G[S]$ and run a \textit{greedy randomized MIS} algorithm, partially
on $G[S]$.
We show that this step can be implemented in just $O(1)$ rounds in the congested clique.
Furthermore, we show that in this step we compute an independent set $I \subseteq S$ of $G[S]$,
such that the set of nodes $R := S \setminus (I \cup N(I))$ that still need to be
processed, induces a subgraph $G[R]$ with maximum degree $O(\poly(\log n))$ w.h.p.
(Here $N(I)$ refers to the union of the neighborhoods in $G[S]$ of 
nodes in $I$.)
In the fourth and final step, we compute an MIS of $G[R]$ using the 
congested clique MIS algorithm of Ghaffari \cite{GhaffariPODC17}
that computes an MIS of a graph in $O(\log \log \Delta)$ rounds, provided the maximum 
degree $\Delta$ is small enough.
Putting these four steps together yields a 2-ruling set algorithm that runs in $O(\log \log \log n)$ w.h.p.

\subsection{Degree-Decomposition Step}
\label{sub:degree-decomposition}
 
Let $G = (V, E)$ be an arbitrary graph. 
Let $U_1$ be the set of all nodes in $G$ with degrees in the range $[n^{1/2}, n)$.
Let $V_1$ be the remaining nodes, i.e., $V_1 = V \setminus U_1$.
Let $U_2$ be the set of all nodes in $V_1$ with degrees in $G[V_1]$ belonging
to the range $[n^{1/4}, n^{1/2})$.
The decomposition continues in this manner until $V$ is partitioned into sets 
$U_1, U_2, \ldots$.
We now provide a more formal description.
For $k = 0, 1, 2, \ldots$, let
$D_k = n^{1 / 2^k}$. The $D_k$'s will serve as degree thresholds and will lead
to a vertex partition. Let $k^* = \lceil \log \log n \rceil$. Note that
$1 < D_{k^*} \leq 2$. Let $V_0 = V$, $G_0 = G$, and
$U_1 = \{v \in V_0 \mid \degree_{G_0}(v) \in [D_1, D_0)\}$. For
$1 \leq k < k^*$, let \[V_k = V_{k-1} \setminus U_k, \qquad G_k = G[V_k],
  \qquad U_{k+1} = \{v \in V_k \mid \degree_{G_k}(v) \in [D_{k+1}, D_k)\}\]
Let $V_{k^*} = V_{k^*-1} \setminus U_{k^*}$, $G_{k^*} = G[V_{k^*}]$, and
  $U_{k^*+1} = V_{k^*}$. 
See Figure~\ref{fig:degree-decomposition} for an illustration of this decomposition.   
Let $N_G(v)$ denote the set of neighbors of vertex $v$ in
  graph $G$. Here are some easy observations:
  \begin{itemize}
    \item[(i)] For $0 \leq k \leq k^*$, $\Delta(G_k) < D_k$.

    \item[(ii)] For $1 \leq k \leq k^*+1$, if $v \in U_k$ then
      $|N_G(v) \cap V_{k-1}| < D_{k-1}$.

    \item[(iii)] For $1 \leq k \leq k^*+1$, if $v \in U_k$ then
      $|N_G(v) \cap U_j| < D_j$ for $j = 1, 2, \ldots k-1$.
  \end{itemize}
\begin{figure}
\begin{boxedminipage}{\textwidth}
\centering
\includegraphics[width=0.9\textwidth]{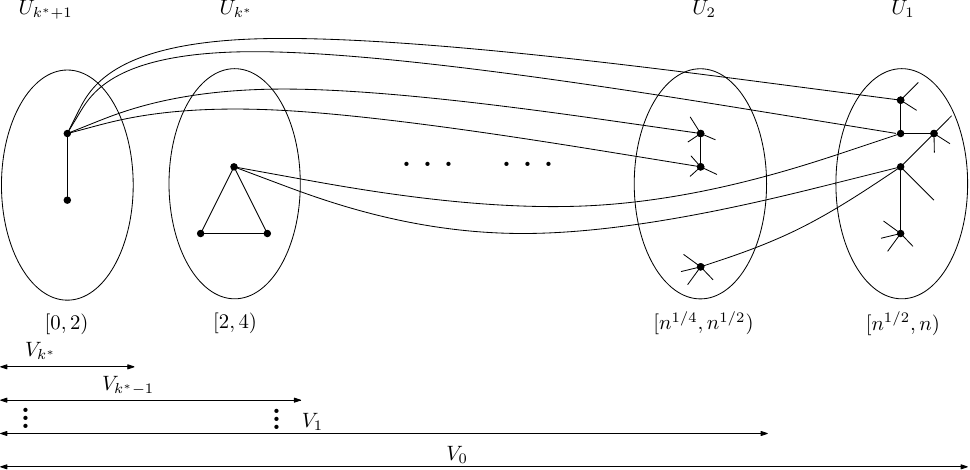} 
\end{boxedminipage}
\caption{Degree-Decomposition Step. 
$U_1$ is the set of all nodes in $G$ with degrees in the range $[n^{1/2}, n)$ and 
$V_1$ is the remaining nodes.
$U_2$ is the set of all nodes in $V_1$ with degrees in $G[V_1]$ belonging
to the range $[n^{1/4}, n^{1/2})$.
The decomposition continues in this manner until all nodes belong to some $U_k$.
We use $k^*$ to denote $\lceil \log\log n\rceil$. Assuming that $\log\log n = k^*$,
we see that $U_k^*$ is the set of nodes that have degree in $G[V_{k^*-1}]$ in the range $[2, 4)$.
Note that a node $v$ that belongs to $U_{k+1}$ could have degree in $G$ that is much
larger than $D_k = n^{1/2^k}$.
\label{fig:degree-decomposition}}
\end{figure}

Now we describe algorithm to compute this degree-decomposition; in particular,
we precisely describe how each node $v$ computes an index $k(v) \in [k^* + 1]$ such that
$v \in U_{k(v)}$.
Below, we first describe at a high level a 2-phase approach that we use to compute the index
$k(v)$ for each vertex $v$. 
Subsequently we will flesh out our approach with necessary details and show that it is correct and can be implemented in $O(\log \log \log n)$ rounds on a congested clique.

\begin{description}
  \item[Lazy phase:] Let $t = \lceil 1 + \log \log \log n \rceil$. The sets
    $U_1, U_2, \ldots, U_t$ are identified in a leisurely manner, one-by-one, in
    $O(\log \log \log n)$ rounds. At the end of this phase each vertex
    $v \in \cup_{i=1}^t U_i$ knows the index $k(v) \in [t]$ such that
    $v \in U_{k(v)}$.

  \item[Speedy phase:] The set of remaining vertices, namely $V_t$, induces a
    graph $G_t$ whose maximum degree is less than
    \[D_t \leq n^{1 / 2^{1 + \log \log \log n}} = n^{1 / (2 \log \log n)}.\]
    This upper bound on the maximum degree helps us compute the index values $k(v)$
    for the remaining vertices at a faster rate. 
    We first show that each vertex $v$
    in $G_t$ can acquire knowledge of the graph induced by the ball $B_{G_t}(v,k^*)$
    in $O(\log \log \log n)$ rounds via a fast \textit{ball-growing algorithm}. (Recall that
    $k^* = \lceil \log \log n \rceil$.) We then show that
    $G[B_{G_t}(v,k^*)]$ contains enough information for $v$ to determine
    $k(v) \in [k^* + 1]$ via local computation. Therefore, after each vertex
    $v \in V_t$ acquires complete knowledge of the radius-$k^*$ ball centered at it,
    it can locally compute index $k(v)$ and proceed to the vertex-selection step.
\end{description}

\noindent We now present the \textit{Lazy-phase algorithm} executed by all
vertices $v \in G$.
\begin{algorithm}[H]
  \caption{Lazy-phase algorithm at vertex $v$}
  \begin{boxedminipage}{\textwidth}
    \small 
    \begin{algorithmic}[1]
      \STATE $k(v) \leftarrow 0$
      \FOR{$i \leftarrow 1$ \TO $t$}
      \STATE $s(v) \leftarrow |\{u \in N_G(v) \mid 1 \leq k(u) < i\}|$
      \IF{$degree_G(v) - s(v) \in [D_i, D_{i-1})$}
	\STATE $k(v) \leftarrow i$
	\STATE Send $k(v)$ to all neighbors
	\STATE \textbf{break}
      \ENDIF
    \ENDFOR
  \end{algorithmic}
\end{boxedminipage}
\end{algorithm}

\begin{lemma}
  The Lazy-phase algorithm runs in $O(\log \log \log n)$ rounds and at the end of
  the algorithm, for each vertex $v \in \cup_{j=1}^t U_j$, $k(v)$ has a value in
  $[t]$ such that $v \in U_{k(v)}$. For any vertex $v \notin \cup_{j=1}^t U_j$,
  $k(v)$ is set to $0$.
\end{lemma}
\begin{proof}
  Given that the sets $U_1, U_2, \ldots, U_i$ have been determined, and that the
  members of each are known to every node in the network, each node can locally
  determine its degree in $G_i = G[V_i]$ and thus determine its membership in
  $U_{i+1}$. Each node can then broadcast whether or not it has joined $U_{i+1}$,
  thus providing knowledge of $U_{i+1}$ to every node in the network. It follows
  that the implementation of the Lazy-phase algorithm requires exactly
  $t = \lceil 1 + \log \log \log n \rceil$ rounds of communication to complete.
\end{proof}

\noindent We now present the \textit{Speedy-phase algorithm} executed by vertex
$v$. Note that the Speedy-phase algorithm is only executed at vertices $v$ for
which $k(v)$ is $0$ after the Lazy-phase algorithm. In other words, the
Speedy-phase algorithm is only executed at vertices $v$ in $G_t$, the graph
induced by vertices not in $\cup_{j=1}^t U_j$.
The key idea of the Speedy-phase algorithm is that once each node $v$ in $G_t$ has acquired knowledge of $G_t[B_{G_t}(v, r)]$, then in constant rounds of communication, each node $v$ can ``double'' its knowledge, i.e., acquire knowledge of $G_t[B_{G_t}(v, 2r)]$. This is done by each node $v$ sending knowledge of $G_t[B_{G_t}(v, r)]$ to all nodes in $B_{G_t}(v, r)$; the key is to establish that this volume of communication can be achieved on a congested clique in constant rounds.
This idea has appeared in a slightly different context in \cite{LenzenWattenhoferBAPODC2010}.

\begin{algorithm}[H]
  \caption{Speedy-phase algorithm at vertex $v$}
  \begin{boxedminipage}{\textwidth}
    \small
    \begin{algorithmic}[1]
      \STATE \COMMENT{Growing the ball $B_{G_t}(v,k^*)$}
      \STATE Each node sends a list of all of its neighbors in $G_t$ to each of
      its neighbors (in $G_t$) \COMMENT{After which each $v \in V_t$ knows
	$G[B_{G_t}(v,1)]$}
	\FOR{$i \leftarrow 0$ \TO $\lceil \log \log \log n \rceil - 1$}
	\label{algo:speedy:ballstart}
	\STATE Send a description of $G[B_{G_t}(v,2^i)]$ to all nodes in
	$B_{G_t}(v,2^i)$
	\STATE Construct $G[B_{G_t}(v, 2^{i+1})]$ from $G[B_{G_t}(u, 2^i)]$ 
	received from all $u \in B_{G_t}(v, 2^i)$
	\label{algo:speedy:ballend}
      \ENDFOR
      \STATE Locally compute $k(v) \in [k^* + 1]$ such that $v \in U_{k(v)}$
    \end{algorithmic}
  \end{boxedminipage}
\end{algorithm}

\begin{lemma}
  The Speedy-phase algorithm above runs in $O(\log \log \log n)$ rounds in the
  congested-clique model and when this algorithm completes execution, each vertex $v$
  in $G_t$ knows $G[B_{G_t}(v,k^*)]$.
\end{lemma}
\begin{proof}
  Line~2 of the Speedy-phase algorithm can be completed in a constant number of rounds using \lra~because each node needs only to send and receive $O(D_t)$ messages to/from $O(D_t)$ neighbors (each message listing a neighbor and destined for a neighbor), as the maximum degree of $G_t$ is less than $D_t$.

  In implementing the Speedy-phase algorithm, the key step is to perform Line~4 in
  $O(1)$ rounds of communication. If this can be done, then after
  $O(\log \log \log n)$ rounds, each node $v$ remaining in $V_t$ will have
  knowledge of its entire neighborhood graph out to a distance of
  $2^{\lceil \log \log \log n \rceil} \geq \lceil \log \log n \rceil = k^*$ hops
  away from $v$.

  Since $G_t$ has maximum degree less than $D_t$, the neighborhood graph
  $G[B_{G_t}(v,2^i)]$ can be completely described by listing all $O(D_t^{2^i+1})$
  edges. Thus, such a neighborhood can be communicated from $v$ to another node
  (in particular, to any other node in $B_{G_t}(v,2^i)$) via
  $O(D_t^{2^i + 1}) = O(n^{(2^i + 1) / 2^t})$ messages of size $O(\log n)$.
  Therefore, to perform a given iteration of Line~4 within the Speedy-phase
  algorithm, each node will need to send (and receive) $O(n^{(2^i + 1) / 2^t})$
  messages (of size $O(\log n)$) to $O(D_t^{2^i}) = O(n^{2^{i-t}})$ other nodes in
  the network. As above, we can use \lra~to perform this
  task in $O(1)$ rounds as long as the total number of messages to be sent (and
  received) by each node is $O(n)$.

  Thus, Line~4 of the Speedy-phase algorithm can be executed in a constant number
  of rounds if $n^{(2^{i+1}+1) / 2^t}$ $= O(n)$; in other words, if
  $2^{i+1} + 1 \leq 2^t$, or $i \leq t - 2 = \lceil \log \log \log n \rceil - 1$.
  This lower bound on the maximum value of $i$ that still allows Line~4 to be
  completed in $O(1)$ rounds is precisely the final index in the for-loop
  (Line~3). This completes the proof.
\end{proof}

\begin{lemma}
  For any graph $H$ and a vertex $v$ in $H$, suppose that $v$ knows the graph
  induced by $B_H(v,k^*)$. Then $v$ can locally compute the index
  $k(v) \in [k^* + 1]$ such that $v \in U_{k(v)}$.
\end{lemma}
\begin{proof}
  The proof is by induction. Whether a vertex $u$ is in $U_1$ is determined by its
  degree in $H$. Since $v$ knows $H[B_H(v,k^*)]$ it can determine via local
  computation which $u \in B_H(v,k^*-1)$ belong to $U_1$ and which don't. As the
  inductive hypothesis, suppose that for some $i \geq 1$, $v$ has determined for
  all $u \in B_H(v,k^*-i)$ the following information:
  \begin{itemize}
    \item[(i)] if $u \in \cup_{j=1}^i U_j$, then $v$ knows $k(u) \in [i]$ such that
      $u \in U_{k(u)}$.
    \item[(ii)] if $u \not\in \cup_{j=1}^i U_j$, then $v$ knows that
      $u \not\in \cup_{j=1}^i U_j$.
  \end{itemize}

  Now consider a vertex $u \in B_H(v,k^*-i-1)$ such that
  $u \not\in \cup_{j=1}^i U_j$. In order to determine if $u \in U_{i+1}$, vertex
  $v$ needs to check if the \textit{residual degree} of $u$, defined as
  \begin{equation}
    \label{residualDegree}
    r(u) := degree_H(u) - |N_H(u) \cap (\cup_{j=1}^i U_j)|
  \end{equation}
  belongs to the interval $[D_{i+1}, D_i)$. In other words, we need to check that
    the degree of $u$ after we have deleted all neighbors in $\cup_{j=1}^i U_j$ is
    in the range $[D_{i+1}, D_i)$. Given the information that $v$ knows about all
      $u \in B(v,k^*-i)$ (by the inductive hypothesis), vertex $v$ can compute the
      residual degree $r(u)$ for each $u \in B_H(v,k^*-i-1)$. Therefore for all such
      $u$, vertex $v$ can determine if $u \in U_{i+1}$ or not. This completes the
      inductive step of the proof.

      Now since $B_H(v,0) = \{v\}$, it follows from the above inductive argument that
      $v$ can determine the index $k(v) \in [k^* + 1]$ such that $v \in U_{k(v)}$.
\end{proof}
\subsection{Vertex-Selection Step}
\label{sub:vertex-selection}
\begin{algorithm}[H]
    \caption{Vertex-Selection Step\label{algo:vertex-selection}}
    \begin{boxedminipage}{\textwidth}
      \small
      \begin{algorithmic}
	\IF{$v \in U_k$ for $k = 1, 2, \ldots, k^*$}  
	\STATE $v$ is selected with probability $\min\left(\frac{2 \log n}{D_k}, 1\right)$
      \ENDIF
      \IF{$v \in U_{k^*+1}$} 
      \STATE $v$ is selected with probability 1
    \ENDIF
  \end{algorithmic}
\end{boxedminipage}
\end{algorithm}
\noindent
As mentioned earlier, the vertex-selection step randomly and independently samples nodes in $G$, with
each node $v$ sampled with a probability $p_v$ that depends on the class $U_{k(v)}$ it belongs to.
Specifically, if $v$ belongs to $U_k$ then $v$ is independently selected with probability $\min(2\log n/D_k,1)$.
Algorithm~\ref{algo:vertex-selection} shows pseudocode for the vertex-selection step.
Let $S$ be the set of vertices that are selected. Let $e(S)$ denote the set of edges in the induced graph $G[S]$.

\begin{lemma} \label{lemma:sizeeS}
With high probability $|e(S)|$ is $O(n \cdot \log^2 n)$.
\end{lemma}
\begin{proof}
  Consider an arbitrary $k$, $1 \leq k \leq k^*+1$ such that $U_k \not= \emptyset$.
  We partition the next part of the proof into two cases depending on how large $D_k$
  is relative to $2 \log n$.

        \begin{itemize}
        \item $D_k < 2 \log n$. In this case, vertices in $U_k$ are selected
                to be in $S$ with probability 1. Since each vertex in $U_k$
                has fewer than $D_{k-1}$ neighbors in $U_k$ and since $D_{k-1} = (D_k)^2 < 4 \log^2 n$,
                the total number of edges in $G[S]$ between vertices in $U_k$
                is at most $O(|U_k| \log^2 n)$.

        \item $D_k  \ge 2 \log n$. In this case, vertices in $U_k$ are selected to
                be in $S$ with probability $2\log n/D_k$. Now there are two cases
                based on the relative sizes of $|U_k|$ and $D_{k-1}$.
                \begin{itemize}
                \item[(a)] If $|U_k| \ge D_k$, then whp each vertex $v \in U_k$ has
                        $$O\left(D_{k-1} \cdot \frac{2 \log n}{D_k} \right) = O(D_k \log n)$$
                neighbors in $U_k \cap S$. Furthermore, whp,
                        $$O\left(\frac{|U_k| \log n}{D_k} \right)$$
                vertices in $U_k$ are selected to be in $S$.
                Therefore, whp there are $O(|U_k| \log^2 n)$ edges among vertices
                in $U_k \cap S$.

                \item[(b)] If $|U_k| < D_k$, then $v$ has $O(D_k)$ neighbors in $U_k \cap S$
                with probability 1. Furthermore, whp $O(\log n)$ vertices in $U_k$
                are selected to be in $S$. Therefore, whp there are $O(|U_k| \log n)$
                edges among vertices in $U_k \cap S$.
                \end{itemize}
        \end{itemize}
        Thus in all cases, the number of edges between vertices in $U_k \cap S$ is $O(|U_k| \log^2 n)$.

        In the input graph $G = (V, E)$, each vertex $v \in U_k$ has fewer than $D_j$ neighbors in $U_j$, for $1 \le j < k$.
  Therefore, whp $v$ has $O(D_j \cdot \log n/D_j) = O(\log n)$ neighbors in $D_j \cap S$. This implies
  that whp $v$ has $O(\log^2 n)$ neighbors in $(\cup_{1 \le j < k} U_j) \cap S$.
  Therefore, whp the total number of edges in $G[S]$ between vertices in $U_k$ and vertices in
  $\cup_{1 \le j \le k} U_j$ is $O(|U_k| \log^2 n)$.

  By summing over all $U_k$, we see that whp the total number of edges in $G[S]$ is
  $O(n \log^2 n)$.
\end{proof}

\begin{lemma}
\label{lemma:highProb}
  For any $v \in V$,
  $\Pr(v \mbox{ is in } S \mbox{ or } v \mbox{ has a neighbor in } S) \geq 1 - 1 / 
  n^2$.
\end{lemma}
\begin{proof}
  Suppose that $v \in U_k$, for some $1 \leq k \leq k^*$. Vertex $v$ has at least
  $D_k$ neighbors in $V_{k-1}$. Each such neighbor is selected for $S$ with
  probability at least $\min \{(2 \log n) / D_k, 1\}$. If $2 \log n \geq D_k$,
  than any of these neighbors is selected for $S$ with probability $1$, so $v$ has
  a neighbor in $S$ with probability $1$. Otherwise, we have
  \[\Pr(v \mbox{ has no neighbor in } S) \leq \left(1 - \frac{2 \log 
  n}{D_k}\right)^{D_k}
\leq e^{-2 \log n} \leq \frac{1}{n^{2.8}}\]
Also, if $v \in U_{k^*+1}$, then $v$ is selected for $S$ with probability $1$.
\end{proof}


\subsection{High Degree Vertex Removal}

Having computed $S \subseteq V$ with the desired properties in $O(\log\log\log n)$ rounds w.h.p., 
we show how to compute an MIS of $G[S]$ in an additional $O(\log\log\log n)$ rounds.
Since by Lemma \ref{lemma:highProb} every vertex in $V$ is, w.h.p., at most 1 hop from a vertex in $S$,
any MIS of $G[S]$ is a 2-ruling set of $G$.
Our first step in computing an MIS of $G[S]$ is to eliminate high degree vertices, so that vertices
that remain have degree $O(\log^3 n)$.
Note that $G[S]$ has an average degree of $O(\log^2 n)$ (by Lemma \ref{lemma:sizeeS}), but
may have much higher maximum degree.

Consider the sequential \textit{Greedy Randomized MIS (GR-MIS)} algorithm (see for e.g., \cite{GhaffariGKMRPODC18})
that starts by randomly permuting vertices and then considers vertices one by one in this permuted order and 
decides greedily if the vertex being considered will join the MIS.
Set $\delta := \log^2 n$. We now show that the GR-MIS algorithm can be partially executed, on vertices with 
ranks $[1 \ldots |S|/\delta]$, in $O(1)$ rounds in the Congested Clique model.

\begin{theorem}
Processing vertices in $S$ with ranks in $[1 \ldots |S|/\delta]$ by the GR-MIS algorithm can be implemented in 
$O(1)$ rounds in the Congested Clique model.
\end{theorem}
\begin{proof}
One vertex (e.g., the one with lowest \texttt{ID}) is designated the leader and it locally generates a random ranking of all vertices in $S$
and tells each vertex in $S$ its rank.
Let $P \subseteq S$ denote the vertices in $S$ with ranks in $[1 \ldots |S|/\delta]$.
Each vertex in $P$ broadcasts a bit, indicating that it is to be processed.
Using this information, each vertex $v \in P$ figures out the set of incident 
edges $E_v$ in $G[S]$ to other vertices in $P$.
Then the plan is to use Lenzen's routing protocol to send all the sets $E_v$,
for all vertices $v \in S$, to the leader.
Using Chernoff bounds, we cam see that w.h.p.~$|E_v|$ is 
$O(n/\delta) = O(n/\log^3 n)$.
To show that Lenzen's routing protocol succeeds in $O(1)$ rounds, we 
need to show that the volume of information
that the leader needs to receive is $O(n)$.
In other words, we need to show that the subgraph $G[P]$ has $O(n)$ edges w.h.p.

Consider an arbitrary vertex $v \in S$. If $degree(v)$ in $G[S]$ is at least $\delta \log n$, then using Chernoff bounds we can see 
that w.h.p.~$v$ has $O(degree(v)/\delta)$ neighbors in $P$.
Therefore the total number of edges in $G[P]$ incident on these ``high degree'' vertices is
$$\sum_{v \in S} O\left(\frac{degree(v)}{\delta}\right) = O\left(\frac{n \log^2 n}{\delta}\right) = O(n).$$

If $degree(v)$ in $G[S]$ is less than $\delta \log n$, then by using
Chernoff bounds we see that w.h.p.~$v$ has $O(\log n)$ neighbors in $P$.
Also, by Chernoff bounds w.h.p.~the number of vertices in $P$ is bounded above by $O(|S|/\delta)$ if $|S| \ge \delta \log n$ and
is bounded above by $O(\log n)$ if $|S| < \delta \log n$.
Therefore, in either case, w.h.p.~the number of vertices in $P$ is bounded above by $O(n/\delta)$.
Therefore, w.h.p., the total number of edges in $G[P]$ incident on
``low degree'' vertices is $O(n/\delta) \cdot O(\log n) = O(n/\log n)$.
Therefore, the total number of edges in $G[P]$ is $O(n)$ w.h.p.

Thus Lenzen's protocol can be used to send $G[P]$ to the
leader in $O(1)$ rounds. The leader locally simulates GR-MIS on $G[P]$ and informs
every vertex in $P$ that has joined the MIS.
Finally, each vertex in $P$ that has joined the MIS broadcasts this
information.
\end{proof}

Let $I \subseteq S$ be the independent set of the vertices selected by the GR-MIS
algorithm on vertices in $S$ with ranks in $[1 \ldots |S|/\delta]$.
Let $R = S \setminus (I \cup N(I))$ be the set of vertices that are remaining to be
processed.

\begin{lemma}
The maximum degree in the graph $G[R]$ is $O(\log^3 n)$ w.h.p.
\label{lemma:lowMaxDegree}
\end{lemma}
\begin{proof}
Lemma 3.1 in \cite{GhaffariGKMRPODC18} implies that w.h.p.~the maximum degree of the graph that remains is
        $$O\left(\frac{|S| \log n}{|S|/\delta}\right) = O(\log^3 n).$$
\end{proof}

\subsection{MIS on Graphs with Low Maximum Degree}
We now compute an MIS on the graph $G[R]$.
To do this we simply use the fact that the maximum degree in $G[R]$ is $O(\log^3 n)$ w.h.p.~(Lemma \ref{lemma:lowMaxDegree}) and appeal
to Lemma 2.15 in \cite{GhaffariPODC17} which asserts that if the maximum degree of a graph $\Delta \le 2^{c' \sqrt{\log n}}$ for a
sufficiently small constant $c'$, then there is a Congested Clique algorithm that computes
an MIS of this graph in $O(\log \log \Delta)$ rounds.
Applying this lemma to $G[R]$ implies MIS can be computed on $G[R]$ in $O(\log\log\log n)$ rounds.

\subsection{Putting it all together}

We now combine the four steps described in the preceeding text: 
(i) degree-decomposition,
(ii) vertex-selection, 
(iii) high-degree vertex removal, and
(iv) MIS computation on graphs with low maximum degree, 
to obtain a 2-ruling set algorithm that runs in $O(\log \log \log n)$ rounds
in w.h.p.
The algorithm is summarized below.

\begin{algorithm}[H]
  \caption{2-Ruling Set Algorithm}
  \begin{boxedminipage}{\textwidth}
    \small 
    \begin{algorithmic}[1]
	    \STATE Run the lazy degree-decomposition algorithm followed by the speedy degree-decomposition algorithm 
      	    \STATE Every vertex $v \in V$ now knows an index $k(v) \in [k^* + 1]$ such that $v \in U_{k(v)}$
	    and we use this knowledge to run the vertex-selection step to compute $S$ 
	    \STATE Run GR-MIS partitally to compute and independent set $I$ of $G[S]$ 
	    \STATE Let $R$ denote $S \setminus (I \cup N(I))$. Run the low-degree MIS algorithm of
	    Ghaffari \cite{GhaffariPODC17} on $G[R]$.
    \end{algorithmic}
  \end{boxedminipage}
\end{algorithm}

\begin{theorem}
There is a 2-ruling set algorithm in the Congested Clique model that runs in $O(\lll)$ rounds
w.h.p.
\end{theorem}
\begin{proof}
	Phases 1, 2, and 5 in the above algorithm take $O(\lll)$ rounds each, w.h.p.
	Phases 3 and 4 take $O(1)$ rounds.
	Every vertex in $V$ is at most 1 hop from $S$, w.h.p. We then compute an MIS of $G[S]$ and
	this MIS is a 2-ruling set of $G$.
\end{proof}

\section{MIS in Growth Bounded Graphs in Constant Rounds}
\label{sec:mis_in_doubling}

Given a metric space $(V,d)$ with constant doubling dimension, we show in this section how to compute an MIS
of a distance-threshold graph $G_r = (V, E_r)$, for any real $r \ge 0$, in 
a \textit{constant} number of rounds on a congested clique.

\subsection{Simulation of the Schneider-Wattenhofer MIS algorithm.}
\label{subsection:SW-MIS}
Before we describe our MIS algorithm, we describe an algorithmic tool that will
prove quite useful.
We know that $G_r$ is growth-bounded and in particular the size of a largest independent set in a ball $B_{G_r}(v, r)$
for any $v \in V$ is $O(r^\rho)$, where $\rho$ is the doubling dimension of $(V, d)$.
Schneider and Wattenhofer \cite{schneider2008logstar} present a deterministic $O(\log^*n)$-round
algorithm to compute an MIS for growth-bounded graphs in the $\cm$ model. 
Suppose that $f$ is a constant such that the Schneider-Wattenhofer algorithms runs in at most
$f \log^*n$ rounds (note that $f$ depends on $\rho$).
We can \textit{simulate} the Schneider-Wattenhofer algorithm in the congested clique model
by (i) having each node $v \in V$ grow a ball of radius $f \log^*n$, i.e., gather a description of the induced graph 
$G[B_{G_r}(v, f\log^* n)]$ and then (ii) having each node $v$ \textit{locally simulate}
the Schneider-Wattenhofer algorithm using the description of $G[B_{G_r}(v, f\log^* n)]$.
Note that since the Schneider-Wattenhofer algorithm takes at most $f \log^* n$ rounds,
it suffices for each node $v \in V$ to know the entire topology of $G[B_{G_r}(v, f\log^* n)]$
to determine if it should join the MIS.
The ``ball growing'' step mentioned above can be implemented by using Lenzen's routing protocol as follows, provided
$\Delta$ (the maximum degree of $G_r$) is not too large.
Each node $v$ can describe its neighborhood using at most $\Delta$ messages of size $O(\log n)$ each.
Node $v$ aims to send each of these $\Delta$ messages to every node $w$ such that $d(v, w) \le r \cdot f\log^*n$.
In other words, $v$ aims to send messages to all nodes in $B_M(v, r \cdot f\log^*n)$.
Since $B_{G_r}(v, f \log^*n) \subseteq B_M(v, r \cdot f\log^*n)$, it follows that the messages sent by $v$ are received
by all nodes in $B_{G_r}(v, f \log^*n)$.
We now bound the size of $B_M(v, r \cdot f\log^*n)$ as follows.
Since $M$ has doubling dimension $\rho$, the size of any MIS in $B_M(v, r \cdot f\log^*n)$ is $O((\log^*n)^{\rho})$ and hence 
total number of nodes in $B_M(v, r \cdot f\log^*n)$ is $O(\Delta \cdot (\log^*n)^{\rho})$. 
Therefore every node $v$ has $O((\log^*n)^{\rho}\cdot \Delta^2 )$ messages to send, each of size $O(\log n)$.
Every node is the receiver of at most $O((\log^*n)^{\rho} \Delta^2)$ messages by 
similar arguments.  
Therefore, if $\Delta = O(\sqrt{n}/(\log^* n)^{\rho/2})$, we can use \lra~to route these messages in $O(1)$ time. We refer 
this simulation of the Schneider-Wattenhofer algorithm \cite{schneider2008logstar} 
as Algorithm \textsc{SW-MIS}. The following theorem summarizes this simulation result. 
\begin{theorem}\label{thm:logstar} 
  If $\Delta(G_r) = O(\sqrt{n}/(\log^* n)^{\rho/2})$ then Algorithm 
  \textsc{SW-MIS} computes an MIS of $G_r$ in $O(1)$ rounds on a congested clique.
\end{theorem}

\subsection{Constant-Round MIS Algorithm}
\label{subsection:constantRoundMIS}
Our MIS algorithm consists of 4 phases.
Next we describe, at a high level, what each phase accomplishes.
\begin{description}
  \item[Phase 1:] We compute vertex-subset $P \subseteq V$ such that (i) every
    vertex in $V$ is at most one hop away from some vertex in $P$ and (ii) $G_r[P]$
    has maximum degree bounded above by $c\cdot\sqrt{n}$, for some constant $c>0$.

  \item[Phase 2:]
    We process the graph $G_r[P]$ and compute two subsets $W$ and $Q$ of $P$
    such that (i) every vertex in $P$ of degree at least $c\cdot n^{1/4}$ is either in
    $W$ or has a neighbor in $W$ and (ii) $Q \subseteq W$ is an independent set
    such that every vertex in $W$ is at most 2 hops from some vertex in $Q$.
    Thus, if we delete $W$ and all neighbors of vertices in $W$ what remains is a graph 
    of maximum degree less than $c \cdot n^{1/4}$.
    Let $V'$ denote the set $P \setminus (W \cup N(W))$. 
    Thus, at the end of Phase 2, $Q$ is a 3-ruling set of $G_r[W \cup N(W)]$
    and $\Delta(G_r[V']) <c\cdot n^{1/4}$. 

  \item[Phase 3:] We compute an MIS $R$ of the graph $G_r[V']$
    by simply calling \textsc{SW-MIS}.

  \item[Phase 4:] Since $Q$ is a 3-ruling set of $G_r[W \cup N(W)]$ and
    $R$ is an MIS of $G_r[V']$,
    we see that $Q \cup R$ is a 3-ruling set of $G_r[P]$ and thus a 4-ruling set of $G_r$. 
    In the final phase, we start with the 4-ruling set $Q \cup R$ and expand this into an MIS 
    $I$ of $G_r$.
\end{description}
Phase 2 is randomized and runs in constant rounds w.h.p. The remaining phases are deterministic and run in constant rounds each.
Algorithm \textsc{LowDimensionalMIS} summarizes our algorithm. 
We now describe each phase in more detail.
\begin{algorithm}[t]
  \caption{\textsc{LowDimensionalMIS}\label{algo:lowDimensionalMIS}}
  \begin{boxedminipage}{\textwidth}
    \small
    \begin{algorithmic}[1]
      \REQUIRE {$G_r = (V, E_r)$}
      \ENSURE {A maximal independent set $I \subseteq V$ of $G_r$}
      \STATE $P \leftarrow \textsc{ReduceDegree}(G_r)$ \COMMENT{Phase 1}
      \STATE $(W, Q) \leftarrow \textsc{SampleAndPrune}(G_r, P)$ \COMMENT{Phase 2}
      \STATE $V' \leftarrow V \setminus (W \cup N(W))$; $R \leftarrow \textsc{SW-MIS}(G_r, V')$ \COMMENT{Phase 3}
      \STATE $S \leftarrow Q \cup R$; $I \leftarrow \textsc{RulingToMIS}(S)$ \COMMENT{Phase 4}
      \RETURN $I$
    \end{algorithmic}
  \end{boxedminipage}
\end{algorithm}

\subsection{Phase 1: Reduce Degree to $O(\sqrt{n})$}
\begin{algorithm}[!ht]
  \caption{\textsc{ReduceDegree} (Phase 1)\label{algo:reduceDegree}}
  \begin{boxedminipage}{\textwidth}
    \small
    \begin{algorithmic}[1]
      \REQUIRE {$G_r = (V,E_r)$}
      \ENSURE {$P \subseteq V$ such that (i) $V = P \cup N(P)$ and (ii) 
	$\Delta(G_r[P]) < c \cdot \sqrt{n}$ for some constant $c > 0$.}
	\STATE Partition $V$ (arbitrarily) into $\lceil\sqrt{n}\rceil$ subsets: $V_1, 
	V_2, \ldots V_{\lceil\sqrt{n}\rceil}$, each of size at most $\sqrt{n}$
	\FORALL{$i \leftarrow 1$ \textbf{to} $\lceil \sqrt{n} \rceil$ \textbf{in 
	parallel}}
	\STATE Send $G_r[V_i]$ to a vertex $v_i$ with lowest ID in $V_i$
	\STATE Vertex $v_i$ executes $P_i \leftarrow \textsc{LocalMIS}(G_r[V_i])$
      \ENDFOR
      \STATE $P \leftarrow \cup_{i=1}^{\lceil\sqrt{n}\rceil} P_i$
      \RETURN $P$
    \end{algorithmic}
  \end{boxedminipage}
\end{algorithm}

\noindent
Algorithm \textsc{ReduceDegree} describes Phase~1 of our algorithm.
The algorithm consists of arbitrarily partitioning the vertex-set of $G_r$ into $\sqrt{n}$ groups of size (roughly) $\sqrt{n}$ each and then separately and in parallel computing an MIS of each part.
Since each part has $\sqrt{n}$ vertices, each part induces a subgraph with at most $n$ edges and 
therefore each such subgraph can be shipped off to a distinct node and 
MIS on each subgraph can be computed locally.
(The subroutine \textsc{LocalMIS} in Line~4 refers to an unspecified MIS algorithm that is executed
locally at a node.)
Using the fact that $G_r$ is growth-bounded, we show that 
the union of all the MIS sets (set $P$, Line~5) induces a graph with maximum degree bounded by $c \cdot \sqrt{n}$
for some constant $c$. Also, we show that Phase~1 runs in constant rounds (Lemma~\ref{lemma:si}).
\begin{lemma}\label{lemma:si}
  Algorithm \textsc{ReduceDegree} completes in $O(1)$ rounds and returns a set
  $P$ such that $\Delta(G_r[P]) < c \cdot n^{1/2}$ for some constant $c > 0$
  (that depends on the doubling dimension of the underlying space).
\end{lemma}
\begin{proof}
  Algorithm \textsc{ReduceDegree} starts by arbitrarily
  partitioning $V$ into $\lceil \sqrt{n} \rceil$ disjoint subsets $V_1, \ldots, 
  V_{\lceil\sqrt{n}\rceil}$ each of size at most $\sqrt{n}$ which can be done in
  $O(1)$ rounds easily.
  Since $|V_i| \le \sqrt{n}$, $G_r[V_i]$ contains at most $n$ edges, for any $i 
  \in \lceil \sqrt{n} \rceil$.
  Using \lra, all knowledge of $G_r[V_i]$ can be shipped off to a designated
  vertex $v_i$ in $V_i$ (e.g., vertex with smallest ID in $V_i$) in $O(1)$ rounds.
  The vertex $v_i$ then computes an MIS $P_i$ of $G_r[V_i]$ locally as shown in
  Line 4 of Algorithm \textsc{ReduceDegree}.
  Finally, $v_i$ informs vertices in $P_i$ of their selection into the MIS.
  The union of the $P_i$'s, denoted $P$, is returned by the algorithm.
  This discussion shows that Algorithm \textsc{ReduceDegree} completes
  in $O(1)$ rounds.

  Consider a vertex $u \in P_i$ for some $i\in [\lceil \sqrt{n} \rceil]$. In
  $G_r[P]$, vertex $u$ cannot have neighbors in $P_i$ since $P_i$ is an
  independent set in $G_r[P]$.
  Consider a set $P_j$, $j\neq i$.
  The distance between any two vertices in $N(u) \cap P_j$ must be more than $r$
  (these nodes are independent) and it must be at most $2r$ (by the triangle
  inequality).
  Since the underlying metric space has doubling dimension $\rho$, it follows that
  $|N(u) \cap P_j| \le 2^{\rho}$.
  Hence the degree of $u$ in $G_r[P]$ is bounded above by $2^{\rho} \cdot (\lceil 
  \sqrt{n} \rceil - 1)$.
  The result follows.
\end{proof}


\subsection{Phase 2: Sample and Prune}
\label{sub:phase2}
\begin{algorithm}[t]
  \caption{\textsc{SampleAndPrune} (Phase 2)\label{algo:sampleAndPrune}}
  \begin{boxedminipage}{\textwidth}
    \small
    \begin{algorithmic}[1]
      \REQUIRE {$(G_r, P)$}
      \ENSURE {$(W, Q)$, $W \subseteq P$ such that $\{v \in P \mid 
	\degree_{G_r[P]}(v) \ge n^{1/4}\} \subseteq W \cup N(W)$; 
      independent set $Q \subseteq W$ such that $Q$ is a 2-ruling set of $G_r[W]$.}
      \FORALL{$v \in P$ \textbf{in parallel}}
      \STATE Vertex $v \in P$ adds itself to $W_i$ with probability $1/n^{1/4}$ 
      for $i = 1, 2, \ldots, \lceil 2 \cdot \log n \rceil$.  
    \ENDFOR
    \STATE $W \leftarrow \cup_{i=1}^{\lceil 2 \log n \rceil} W_i$
    \FORALL{$i \leftarrow 1$ \textbf{to} $\lceil 2 \log n \rceil$ \textbf{in 
    parallel}}
    \STATE Send $G_r[W_i]$ to a vertex $w_i$, where $w_i$ is the vertex of 
    rank $i$ in the sequence of vertices in $V$ sorted by increasing ID
    \STATE Vertex $w_i$ executes $X_i \leftarrow \textsc{LocalMIS}(G_r[W_i])$
  \ENDFOR
  \STATE $Q \leftarrow \textsc{SW-MIS}(G_r[\cup_{i=1}^{\lceil 2 \log n \rceil} 
  X_i])$
  \RETURN $(W, Q)$
\end{algorithmic}
\end{boxedminipage}
\end{algorithm}

\noindent
Algorithm \textsc{SampleAndPrune} implements Phase 2 of our MIS algorithm.
It takes the induced subgraph $G_r[P]$ as input and starts by computing a set $W 
\subseteq P$ using a simple random sampling approach.
Specifically, for each $i = 1, 2, \ldots, \lceil 2 \cdot \log n \rceil$, each vertex in $P$ simply 
adds itself to a set $W_i$ independently, with probability $1/n^{1/4}$.
We start by proving a useful property of $W$.

\begin{lemma}\label{lemma:nice}
  Every node $u$ with degree at least $n^{1/4}$ in $G_r[P]$ has a neighbor in $W$ with probability at least
  $1-\frac{1}{n^2}$. 
\end{lemma}
\begin{proof}
  Let $u \in P$ be a node with degree at least $n^{1/4}$ in $G_r[P]$.
  For any neighbor $v$ of $u$, $\Pr(v \notin W) \leq \left(1 - \frac{1}{n^{1/4}}\right)^{\lceil 2\log n \rceil}$.
  Therefore the probability that no neighbor of $u$ is in $W$ is at most
  $\left(1-\frac{1}{n^{1/4}}\right)^{\lceil 2\log n \rceil \cdot n^{1/4}}$.
  This is bounded above $e^{-\lceil 2\log n \rceil}$, which is bounded above by $1/n^2$.
\end{proof}

After using random sampling to compute $W$, Algorithm \textsc{SampleAndPrune} 
then ``prunes'' $W$ in constant rounds to construct a subset $Q \subseteq W$ such 
that $Q$ is a 2-ruling set of $W$. In the rest of this subsection we prove that 
Algorithm \textsc{SampleAndPrune} 
does behave as claimed here.

\begin{lemma} \label{lemma:edgew}
  The number of edges in $G_r[W_i]$ is $O(n)$ w.h.p., for each $i  = 1, 2, \ldots, \lceil 2 \log n \rceil$.  
\end{lemma}
\begin{proof}
We first bound the size of the set $W_i$ and the maximum degree of $G_r[W_i]$ for any $i = 1, 2, \ldots, \lceil 2 \log n \rceil$. 
Observe that $\E[|W_i|] = n^{3/4}$ and since nodes join $W_i$ independently,
an application of Chernoff's bound~\cite{dubhashiBook} yields 
$\Pr(|W_i| \leq 6n^{3/4}) \geq 1 - \frac{1}{n^2}$. 
To bound $\Delta(G_r[W_i])$ we use the fact that degree of any node in $G_r[P]$ is at most $\sqrt{n}$ and therefore the expected degree of any node in 
$G_r[W_i]$ is at most $n^{1/4}$.
Another application of Chernoff's bound yields $\Pr(\mbox{degree}_{G_r[W_i]}(v) \leq 6n^{1/4}) \geq 1 - \frac{1}{n^2}$ for each node $v$.
Using the union bound over all nodes $v \in W_i$ yields that with probability at least $1 - 
\frac{1}{n}$ every node in $G_r[W_i]$ has degree at most $6n^{1/4}$.
Hence, with high probability, the number of edges in $G[W_i]$ is at most $36n$. 
\end{proof}

\begin{lemma}\label{lemma:wimis}
  The set $ X := \cup_{i=1}^{\lceil 2\log n \rceil} X_i \subseteq P$ is 
  computed in constant rounds w.h.p.~in Lines 4-6
  of Algorithm \textsc{SampleAndPrune}. 
  Furthermore, Every vertex in $W$ is at most one hop away from some vertex in $X$.
\end{lemma}
\begin{proof}
  We argue that Line~5 can be implemented in $O(1)$ rounds w.h.p.
  By Lemma~\ref{lemma:edgew}, each node has to send at most $O(n^{1/4})$ messages to $w_i$ and w.h.p.~each$w_i$ receives at most $O(n)$ messages. Therefore by \lra\ Line 5 takes $O(1)$ rounds. 
  To repeat this for each $i  = 1, 2, \ldots, \lceil 2\log n \rceil$ in parallel, 
  every node 
  has to send at the most $\lceil 2\log n \rceil \cdot n^{1/4}$ messages. 
  Since $w_i$'s are distinct no $w_i$ needs to receive more than $O(n)$ 
  messages.   

  Each $v \in W$ belongs to $W_i$ for some $i$ and is therefore at 
  most one hop from some vertex in $X_i$.
\end{proof}

\begin{lemma}\label{lemma:qrule}
  W.h.p.~it takes constant number of rounds to compute $Q$.
  Furthermore, $Q$ is a 2-ruling set of $G_r[W]$.
\end{lemma}
\begin{proof}
  Consider a node $v \in \cup_{i=1}^{\lceil 2 \log n\rceil}  {X_i}$.
  Since each $X_i$ is an independent set, by using \aspect\  
  of $G_r[X_i]$, we see that the number of neighbors of $v$ in $X_i$ is bounded
  above by a constant. Hence, the maximum degree in
  $G_r \left[\cup_{i=1}^{\lceil 2 \log n\rceil}  {X_i}\right]$ is $O(\log n)$. 
  Since the maximum degree of this growth-bounded graph is $O(\log n)$, by 
  Theorem~\ref{thm:logstar} an MIS of this graph can be computed in constant 
  rounds by using \textsc{SW-MIS}.

  A node $v \in W$ belongs to some $W_i$ and is therefore at most one hop from
  some node in $X_i$. Also, every node in every $X_i$ is at most one hop from some
  node in $Q$. Also, $Q$ is independent and therefore $Q$ is a 2-ruling set of
  $G_r[W]$.
\end{proof}

\subsection{Phase 4: Ruling Set to MIS}
\label{sub:phase4}
Algorithm \textsc{RulingToMIS} implements Phase 4 of our MIS algorithm. The 
algorithm takes as input the graph $G_r$ and the vertex subset  
$S = Q \cup R$ where $Q$ and $R$ are the outputs of Phase~2 and 
Phase~3, respectively. 
Note that Lemma~\ref{lemma:qrule} implies that $S$ is a 4-ruling set of $G_r$. 
This property is used to cover $G_r$ with balls of radius $4r$, 
centered at members of $S$. 

Consider the graph $G_{9r} = (V, E_{9r})$ where $E_{9r} = \{\{u,v\} \mid u, 
v \in V \mbox{ and } d(u, v) \leq 9r\}$.
In Lemma~\ref{lemma:coloring} we prove a constant upper bound on the maximum degree 
$\Delta(G_{9r}[S])$. This allows us to compute a proper vertex coloring of $G_{9r}[S]$
using a constant number of colors.
This coloring guides the rest of the algorithm, providing a  
schedule for processing the vertices in the aforementioned balls centered at vertices in $S$. 
For each color $i$, the algorithm processes all vertices in $S$ colored $i$ in parallel.
For each vertex $v \in S$ colored $i$, let $B_v$ denote the subset of $B(v, 4r)$ of vertices
still ``active''. 
The algorithm computes an MIS of the induced subgraph $G_r[B_v]$; this computation occurs in
parallel for each $v$ colored $i$.
Since the vertex coloring is with respect to $G_{9r}$, two balls $B_v$ and $B_{v'}$ that are
processed in parallel do not intersect and in fact are not even connected by an edge.
Thus processing in parallel all of the balls $B_v$ for $v$ colored $i$ has no untoward
consequences. 
We note that due to the growth bounded property, every independent set of $G_r[B_v]$ has a 
constant number of vertices.
Hence, we can use a simple sequential algorithm to compute an MIS of $G_r[B_v]$ -- repeatedly
each vertex with smallest ID in its neighborhood joins the MIS and the graph is updated.
We call this MIS algorithm \textsc{SequentialMIS} and use it in Line~9 in Algorithm 
\textsc{RulingToMIS}.
Since every vertex in $V$ is at distance at most $4r$ from some vertex in $S$, every vertex
in $V$ is is some ball $B_v$ and is eventually processed.


\begin{algorithm}[t]
  \caption{\textsc{RulingToMIS} (Phase 4) \label{algo:rulingMIS}}
  \begin{boxedminipage}{\textwidth}
    \small
    \begin{algorithmic}[1]
      \REQUIRE $(G_r, S=Q\cup R)$ 
      \ENSURE A maximal independent set $I \subseteq V$ of $G_r$
      \STATE $E_{9r} \leftarrow \left\{\left\{u, v\right\} \mid u, v \in S 
      \mbox{ and } d(u,v) \leq 9r\right\}$ 
      \STATE $G_{9r}[S] \leftarrow (S, E_{9r})$
      \STATE Send $G_{9r}[S]$ to a vertex $v^*$ with lowest ID in $S$
      \STATE Vertex $v^*$ executes $\Psi \leftarrow 
      \textsc{LocalColoring}(G_{9r}[S])$
      with color pallet $\{1,2,\ldots, \gamma+1\}$. Here
      $\gamma$ is the constant from Lemma~\ref{lemma:coloring}. 
      \STATE $V' \leftarrow V$
      \FOR{$i = 1$ \TO $i = \gamma+1$}
      \FORALL{$v \in S$ such that $\Psi(v) = i$ \textbf{in parallel}}
      \STATE $B_v \leftarrow \left\{u \mid u \in V' \mbox{ and } d(u, v) \leq 
    4r\right\}$
    \STATE $I_v \leftarrow \textsc{SequentialMIS}(G_{9r}[B_v])$
  \ENDFOR
  \STATE $V' \leftarrow V'\setminus \left( \cup_{v \in S \wedge 
  \Psi(v) = i}\left( N(I_v) \right)\right)$ 
\ENDFOR
\STATE $I \leftarrow \union_{v\in S} I_v$
\RETURN $I$
\end{algorithmic}
 \end{boxedminipage}
\end{algorithm}

\begin{lemma}\label{lemma:coloring}
  $\Delta(G_{9r}[S]) \leq \gamma$, where $\gamma$ is a constant. 
\end{lemma}
\begin{proof}
  Consider any node $v \in S$ and neighbors $N_{G_{9r}}(v)$ of $v$
  in $G_{9r}[S]$. By the triangle inequality, any pair of nodes in 
  $N_{G_{9r}}(v)$ are at most distance $18r$ apart and by 
  Lemma~\ref{lemma:qrule}, at least distance $r$ 
  apart. Hence $N_{G_{9r}}(v) \cup v$ has a constant aspect
  ratio and by \aspect, we have $|N_{G_{9r}}(v) \cup v| \leq 18^{\rho} = 
  \gamma$.
\end{proof}
\begin{lemma}\label{lemma:constantRounds}
Algorithm \textsc{RulingToMIS} executes in a constant number of rounds.
\end{lemma}
\begin{proof}
  Since the maximum degree of $G_{9r}[S]$ is a constant, the entire description 
  of $G_{9r}[S]$ can be shipped to a designated vertex $v^*$ (e.g., a vertex with 
  the smallest ID) using \lra~in $O(1)$ rounds. Then $v^*$ can compute a coloring 
  of $G_{9r}[S]$ such that no two adjacent vertices have the same color. Notice 
  that the maximum degree of $G_{9r}[S]$ is bounded above by $\gamma$, hence 
  $\gamma+1$ colors are sufficient. 

  The constant upper bound on the size of the color palette implies that the 
  \textbf{for}-loop starting in Line~6 executes a constant number of iterations.
  In each iteration $i$, all nodes $v \in S$ colored $i$ are processed. Specifically, 
  an MIS of $G_r[B_v]$ is computed and since the size of every independent set in $G_r[B_v]$
  is bounded above by a constant (by appealing to \aspect), Algorithm \textsc{SequentialMIS}
  terminates in constant rounds. Hence, each iteration of the outer-\textbf{for}-loop takes
  a constant number of rounds of communication.
\end{proof}

\begin{lemma}\label{lemma:correctness}
The set $I$ computed by Algorithm \textsc{RulingToMIS} is an MIS of $G_r$.
\end{lemma}
\begin{proof}
First we show that $I$ is an independent set by contradiction.
Suppose that for some $p, q \in I$, $p$ and $q$ are adjacent in $G_r$.
Then it must be the case that both $p$ and $q$ were selected in the same
iteration of the outer-\textbf{for}-loop; otherwise, the selection of one
of the two nodes would render the other unavailable for selection.
If $p$ and $q$ are selected in the same outer-\textbf{for}-loop iteration,
it must be the case that $p \in B_v$ and $q \in B_{v'}$ where $v \not= v'$, but
$v$ and $v'$ have the same color.
Since $d(p, v) \le 4r$, $d(q, v') \le 4r$, and $d(p, q) \le r$, using the triangle
inequality we see that $d(v, v') \le 9r$.
But, if this is the case then there is an edge between $v$ and $v'$ in $G_{9r}[S]$
and these two vertices would not have the same color, contradicting our earlier
conclusion that $v$ and $v'$ have the same color.

We now prove that $I$ is maximal.
Since $S$ is a 4-ruling set of $G_r$, every node $u \in V$ is in $B(v, 4r)$ for
some $v \in S$.
Suppose that $v$ is colored $i$ and therefore $B_v$ is processed in iteration $i$ of
the outer-\textbf{for}-loop.
If $u \in B_v$ then Algorithm \textsc{SequentialMIS} will either pick $u$ or a neighbor
to join the MIS.
Otherwise, if $u \not\in B_v$ then it must be the case that in an earlier iteration
of the outer-\textbf{for}-loop, either $u$ or a neighbor were selected to be in
the MIS.
\end{proof}

\section{Constant-Approximation to MST in Constant Rounds} 
\label{sec:mst}

For a metric space $(V,d)$, define a \textit{metric graph} $G = (V, E)$ as the clique on set $V$ with each edge $\left\{u,v\right\}$ 
having weight $d(u,v)$. 
In this section we present a constant-round algorithm for computing a constant-factor 
approximation of an MST of given metric graph $G = (V, E)$
with constant doubling dimension.
We require that at the end of the MST algorithm, each node in $V$ know the entire spanning tree.
Our overall approach is as follows.
We start by showing how to ``sparsify'' $G$ and construct a spanning subgraph $\Ghat = (V, \ehat)$, $\ehat \subseteq E$, such that
$wt(MST(\Ghat)) = O(wt(MST(G)))$. Thus computing an MST on $\Ghat$ yields an $O(1)$-approximation to an MST on $G$.
The sparsification is achieved via the construction of a collection of maximal
independent sets (MIS) \textit{in parallel} on different distance-threshold subgraphs of $G$.
Thus we have reduced the problem of constructing a constant-approximation of an MST on the metric graph $G$ to two problems:
(i) the MIS problem on distance-threshold graphs and (ii) the problem of computing an MST of 
a sparse graph $\Ghat$.
Using the fact that the underlying metric space $(V, d)$ has constant doubling dimension, we show that $\Ghat$ has
linear (in $|V|$) number of edges. As a result, problem (ii) can be easily solved in constant number of rounds by simply shipping $\Ghat$ to a single node for local MST computation.
In Section~\ref{sec:mis_in_doubling}, we have already shown how to compute an MIS of a 
distance-threshold graph 
in a constant doubling dimensional space on a congested clique in constant number of rounds.
Finally, we show that due to the particular bandwidth usage of our MIS algorithm, we can run all
of the requisite MIS computations in parallel in constant rounds.


\subsection{MST Algorithm}
\label{subsection:algorithm}

We now present our algorithm in detail; the reader is encouraged to follow along the pseudocode in Algorithm \ref{algo:construct}.
We partition the edge set $E$ of the metric graph into two subsets $E_\ell$ (\textit{light} edges) and
$E_\ehh$ (\textit{heavy} edges) as follows. Let $d_m = \max\left\{d(u,v) \mid
  \left\{u,v\right\} \in E  \right\}$ denote the diameter of the metric space
\footnote{If the size of the encoding of distances is more than $O(\log n)$ bits then it is suffices to know only most-significant $\log n$-bits of encoding of $d_m$ to act as ``proxy'' for $d_m$ which will only increase the approximation factor by a constant.}.  
  Define $E_\ell  =  \left\{\left\{u,v\right\} \mid d(u,v) \leq 
d_m/n^3\right\} $ and $E_\ehh =  E \setminus E_\ell$.
We deal with these two subsets $E_\ell$ and $E_\ehh$ separately. 
\begin{algorithm}[t]
  \caption{\textsc{MST-Approximation} \label{algo:construct}}
  \begin{boxedminipage}{\textwidth}
    \small
    \begin{algorithmic}[1]
      \REQUIRE A metric graph $G=(V,E)$ on metric space $(V,d)$
      \ENSURE A tree $\cthat$ such that $wt(\cthat) = 
      O\left(wt\left(MST\left(G\right)\right)\right)$
      \STATE $d_m = \max\{d(u, v) \mid \{u, v\} \in E\}$
      \STATE $E_\ell \leftarrow \left\{\left\{u,v\right\} \mid d(u,v) \leq \frac{d_m}{n^3}\right\}$ \COMMENT{Processing light edges}
      \STATE $S \leftarrow $ \textsc{ComputeMIS}$(G[E_0])$ where $E_0 \leftarrow \left\{\left\{u,v\right\} \mid d(u,v) \leq \frac{d_m}{n^2}\right\}$
      \STATE $\ehat_\ell \leftarrow \left\{\left\{u,v\right\} \mid u \in S \mbox{ and } d(u,v) \leq \frac{2\cdot d_m}{n^2}\right\}$ 
      \STATE $E_\ehh \leftarrow \left\{\left\{u,v\right\} \mid d(u, v) > \frac{d_m}{n^3}\right\}$ \COMMENT{Processing heavy edges}
      \STATE $h \leftarrow \left\lceil \frac{3\log n}{\log c_1} \right\rceil$; $r_0 \leftarrow \frac{d_m}{c_1^h}$
      \FOR{$i=1$ \TO $h$ \textbf{in parallel}} \label{algo:construct:forstart}
      \STATE $r_i \leftarrow (c_1)^i \cdot r_0$
      \STATE $E_i \leftarrow \left\{\left\{u,v\right\} \mid d(u,v) \leq r_i\right\}$
    \STATE $V_i \leftarrow $\textsc{ComputeMIS}$(G[E_i])$
    \STATE $\ehat_i \leftarrow \left\{\left\{u,v\right\} \mid u, v\in V_i \mbox{ 
    and } d(u,v) \leq c_2\cdot r_i\right\}$ 
  \ENDFOR \label{algo:construct:forend}
  \STATE $\ehat_\ehh \leftarrow \cup_{i=1}^{h} \ehat_i$; $\ehat \leftarrow \ehat_\ell \cup \ehat_\ehh$ 
  \RETURN \textsc{MST-Sparse}$(G[\ehat])$ \label{algo:construct:sparsemst} 
\end{algorithmic}
  \end{boxedminipage}
\end{algorithm}

First consider the set of light edges $E_\ell$ and note that
$G[E_\ell]$ may have several components. We would like to select an edge set 
$\ehat_\ell$ such that \begin{inparaenum}[(i)] 
\item any pair of vertices that are in the same connected component in 
  $G[E_\ell]$ are also in the same connected component in $G[\ehat_\ell]$, and 
\item $wt(\ehat_\ell) = O(wt(MST(G)))$. \end{inparaenum} 
(Note that one can define $\ehat_\ell = E_\ell$ to have these two properties but 
we want to ``sparsify'' $E_\ell$, ideally we would like to have $|\ehat_\ell| = O(n)$ and we show this for metric with constant doubling dimension.) 
The algorithm for selecting $\ehat_\ell$ is as follows.
Let $S$ be an MIS of the distance-threshold graph $G_r$,
where $r = d_m/n^2$.
(This MIS computation is not on graph induced by $E_\ell$, notice the $r$. 
This is done to obtain certain properties of $\ehat_\ell$ described above.)  
Define $\ehat_\ell = \left\{\left\{u,v\right\} \mid u \in S \mbox{ and } d(u,v) \leq 
  2\cdot d_m/n^2 \right\}$. 
Note that $\ehat_\ell$ may not be a subset of $E_\ell$. 

Now we consider the set $E_\ehh$ of heavy edges. 
Let $c_1 > 1$ be a constant.  
Let $h$ be the smallest positive integer such that $c_1^h \geq n^3$.
Observe that $h = \left\lceil \frac{3\log n}{\log c_1} \right\rceil$.
Let $r_0 = {d_m}/{c_1^h}$ (note that for any heavy edge $\{u, v\}$, $d(u, v) > r_0$) and let $r_i = c_1 \cdot r_{i-1}$, for $i>0$. 
We construct $\ehat_\ehh$ in \textit{layers} as follows. 
Let $V_0 = V$ and $V_i$ for $0<i\leq h$ is an MIS of the subgraph $G[E_i]$ where $E_i = \left\{\left\{u,v\right\} \mid d(u,v) \leq r_i\right\}$. 
Let $c_2 > c_1+2$ be a constant. 
Define $\ehat_i$, the edge set at the layer $i$ as: 
$\ehat_i = \left\{\left\{u,v\right\} \mid u,v \in V_i \mbox{ and } d(u,v) \leq 
c_2\cdot r_i\right\}$.
We define $\ehat_\ehh = \cup_{i=1}^h \ehat_i$ and $\ehat = \ehat_\ehh \cup
\ehat_\ell$. 
A key feature of our algorithm is that a layer $\ehat_i$ does not depend on
other layers and therefore these layers can be constructed in parallel.
We then call an as-yet-unspecified algorithm called \textsc{MST-Sparse} 
that quickly computes an exact MST of $\Ghat = G[\ehat]$ in the congested clique model.



In the analysis that follows, we separately analyze the processing of light edges and heavy edges.
We first show the \textit{constant-approximation property} of $\Ghat$ which doesn't require metric to be of constant doubling dimension. Later we show if the underlying metric has constant doubling dimension then Algorithm~\ref{algo:construct} runs in constant rounds w.h.p.. 

\subsection{Constant-Approximation Property}\label{sub:weight}
Let $\ct$ be an MST of graph $G=(V,E)$. Let $\cthat$ be a MST of the graph 
$\Ghat = (V,\ehat)$. 
We now prove that $wt(\cthat)=O(wt(\ct))$. 
First we claim that the connectivity that edges in $E_\ell$ (i.e., the light edges)
provide is preserved by the edges selected into $\ehat_\ell$ (Lemma~\ref{lemma:connectivity})
and the total weight of these selected edges is not too high (Lemma~\ref{lemma:wtsmall}). 
Later we prove a similar claim for heavy edges (Lemma~\ref{lemma:wtbig}). 
\begin{lemma}\label{lemma:connectivity}
  For any vertices $s$ and $t$ in $V$, if
  there is a $s$-$t$ path in $G[E_\ell]$ then there exists an $s$-$t$ path
  in $G[\ehat_\ell]$.
\end{lemma}
\begin{proof}
  Consider an edge $\{u, v\} \in E_\ell$.
  If $\left\{u,v\right\} \in \ehat_\ell$ then we are done. If
  $\left\{u,v\right\} \notin \ehat_\ell$ then we show that
  there exists a vertex $w$ such that $\{u, w\}, \{v, w\}
  \in \ehat_\ell$.
  Since $\{u, v\} \in E_\ell$, $d(u, v) \le d_m/n^3$.
  Furthermore, since $\{u,v\} \notin \ehat_\ell$ it means both $u$ and $v$ are not
  in $S$, an MIS of $G_r$, $r = d_m/n^2$.
  Hence there is a vertex $w\in S$ such that
  $d(u, w) \leq d_m/n^2$.
  By the definition of $\ehat_\ell$, $\{u, w\} \in \ehat_\ell$.
  By the triangle inequality, we have $d(v, w) \leq
  d_m/n^2 + d_m/n^3$ which implies $\{v, w\} \in \ehat_\ell$.
  The lemma follows by repeatedly applying above result to each edge of the given $s$-$t$ path.
\end{proof}
\begin{lemma}\label{lemma:wtsmall}
  $wt(\ehat_\ell) = O(wt(\ct))$.
\end{lemma}
\begin{proof}
  The weight of each edge in $\ehat$ is at most $2d_m/n^2$ and since there are
  at most $n^2$ edges in $\ehat_\ell$ (trivially), we see that $wt(\ehat_\ell) = O(d_m)$.
  We obtain the lemma by using the fact that the total weight of any spanning tree is bounded below by $d_m$.
\end{proof}

Consider an edge $\left\{u,v\right\} \in E(\ct)$. Let  
$C(u)$ and $C(v)$ be the components containing $u$ and $v$ respectively 
in the graph $\ct \setminus \left\{u, v\right\}$. 

\begin{lemma}\label{lemma:wtbig}
  If $\left\{u,v\right\} \in E(\ct) \cap E_\ehh$ then 
  there exists an edge $\{u',v'\} \in \ehat$ such that (i) $d(u', v') \le c_2 \cdot d(u, v)$ 
  and (ii) $u' \in C(u)$ and $v' \in C(v)$.
\end{lemma}
\begin{proof}
  Let $i$ be the largest integer such that $r_i < d(u, v)$.
  Hence $d(u, v) \le r_{i+1} =  c_1 \cdot r_i \le (c_2-2) \cdot r_i$ (since
  $c_2$ was chosen to be greater than $c_1+2$).

  Let $u'$ and $v'$ be the nearest nodes in the MIS $V_i$ of $G[E_i]$ 
  from $u$ and $v$ respectively. Note that $u'$ could be $u$ and $v'$ could be 
  $v$.
  Thus $d(u,u') \le  r_i$ and  $d(v,v') \le r_i$.
  By the triangle inequality we have, 
  $d(u',v')  \leq  d(u',u) + d(u,v) + d(v,v') \leq  r_i + (c_2 - 2) \cdot r_i + r_i \le c_2\cdot r_i < c_2 \cdot d(u, v).$
  Hence, $(u',v') \in \ehat_i$ and also note that $d(u',v') \leq \alpha\cdot 
  d(u,v)$ where $\alpha$ is any constant greater than $c_2$. 
  Now note that $\{u, v\}$ is the lightest edge between a vertex in $C(u)$ and a vertex in $C(v)$ by 
  virtue of being an MST edge. Therefore, it is the case that $u' \in C(u)$ and $v' \in C(v)$ since 
  $d(u, u') < d(u, v)$ and $d(v, v') < d(u, v)$.
\end{proof}

\noindent
This lemma implies that for every cut $(X, Y)$ of $G$ and an MST edge $\{u, v\}$ that crosses
the cut, there is an edge $\{u', v'\}$ in $\Ghat$ also crossing cut $(X, Y)$ with weight
within a constant factor of the weight of $\{u, v\}$.
The following result follows from this observation and properties of $\ehat_\ell$ proved earlier.
\begin{theorem}\label{thm:weight}
  Algorithm~\ref{algo:construct} computes a spanning tree $\cthat$ of $G$ such that 
  $wt(\cthat) = O\left(wt\left(MST\left(G\right)\right)\right)$. 
\end{theorem}

\subsection{Constant Running Time}
The result of the previous subsection does not require that the underlying metric space $(V, d)$ 
have constant doubling dimension.
Now we assume that $(V, d)$ has constant doubling dimension and in this setting 
we show that Algorithm \textsc{MST-Approximation} can be implemented in \textit{constant} rounds.
Even though the algorithm is described in  a ``sequential'' style in Algorithm \ref{algo:construct},
it is easy to verify that most of the steps can be easily implemented in constant rounds in
the congested clique model.
However, to finish the analysis we need to show:
(i) that \textsc{ComputeMIS} executes in constant rounds, 
(ii) that the $h = O(\log n)$ calls to \textsc{ComputeMIS} in Line 10 can be executed in parallel in
constant rounds, and (iii) that \textsc{MST-Sparse} in Line 13 can be implemented in constant
rounds. 
In the following, we show (iii) by simply showing that $\Ghat$ has linear number of edges.
In the previous section, we have shown (i) and later in this section we show (ii). 

%

We first show $|\ehat_\ell| = O(n)$ in Lemma~\ref{lemma:smallsize} and then argue about heavy edges. 
\begin{lemma}\label{lemma:smallsize}
  $|\ehat_\ell| = O(n)$.
\end{lemma}
\begin{proof}
  For any edge $\left\{u,v\right\} \in \ehat_\ell$ either $u$ or $v$ or both
  belong to $S$ (by construction). We orient edges such that an edge is directed
  towards the node in $S$. If both end points are in $S$ then we add two
  oppositely directed edges. We prove that the out-degree of a node is bounded
  by a constant.

  \noindent Consider a node $u$. Let $N_o(u)$ be the set of endpoints of all
  outgoing edges of $u$. If $|N_o(u)| < 2$ then we are done, therefore
  consider the case $|N_o(u)| \geq 2$. Consider any two nodes $v_i, v_j \in
  N_o(u)$.  By construction we have, $d(u,v_i) \leq 2\cdot d_m/n^2$
  and $d(u,v_j) \leq 2\cdot d_m/n^2$. Therefore by the triangle inequality,
  $d(v_i, v_j) \leq 4\cdot d_m/n^2$. Also, by the definition of orientation $v_i, v_j
  \in S$ and therefore by the definition of $S$ we have, $d(v_i, v_j) > d_m/n^2$.
  Hence the aspect ratio of $N_o(u)$ is at most $4$.
  By \aspect, we have $|N_o(u)| = O(1)$. Hence, $|\ehat_\ell| = O(n)$.
\end{proof}

Now we show $|\ehat_\ehh| = O(n)$. 
We first show in the following lemma two useful properties of vertex-neighborhoods in the graph induced by $\ehat_i$.

\begin{lemma} \label{lemma:neighborhood}
  For each $u \in V_i$, (i) $|N_i(u)| \leq c_3$ where $c_3={c_2}^{O(\rho)}$  and (ii)
  $N_i(u) \cup \{u\}$ induces a clique in $G[E_j]$ for all $i > 0$ and $ j \geq 
  i + \delta$ where $\delta = \left\lceil \frac{\log 2c_2}{\log c_1} 
  \right\rceil$. 
\end{lemma}
\begin{proof}
  We first show that the aspect ratio of $N_i(u)$ is bounded by $2c_2$.
  This follows from two facts:
  (a) any two points in $N_i(u)$ are at least distance $r_i$ apart,
  and (b) any point in  $N_i(u)$ is at distance at most $c_2\cdot r_i$ from $u$
  and therefore, by using the triangle inequality, any two points in $N_i(u)$ are at
  most $2c_2\cdot r_i$ apart. Then using the bound from
  \aspect~we obtain the result claimed in part (i).

  Now we show part (ii) of the claim.
  If $|N_i(u)| = 0$ then we are done. 
  If $|N_i(u)| = 1 $ then let $v \in N_i(u)$. 
  This implies $d(u,v) \leq c_2\cdot r_i \ < c_1^\delta \cdot r_i = r_{i+\delta}$ 
  which implies $\left\{u,v\right\} \in E_j, j\geq i+\delta$.

  \noindent Now assume $|N_i(u)| > 1$. 
  Consider any two distinct vertices $v, w \in N_i(u)$. 
  Since $\left\{u,v\right\}, \{u,w\} \in \ehat_i$ we have 
  $d(u,v) \leq c_2\cdot r_i$ and $d(u,w) \leq c_2\cdot r_i$.
  By the triangle inequality, $ d(v,w) \leq 2c_2\cdot r_i \leq
  c_1^\delta \cdot r_i =c_{i+\delta}$.
  Therefore $\{v,w\} \in E_{i+\delta}$ and hence we have
  $\{v,w\} \in E_j, \mbox{ for all } j \geq i + \delta$.
\end{proof}

The implication of the above result is that $|\ehat_i|$ is linear in size.
Since we use $O(\log n)$ layers in the algorithm, it immediately follows that $|\ehat_\ehh|$ is $O(n\log n)$.
However, part (ii) of the above result implies that only one of the nodes in $N_i(u)$ will be present in 
$V_j$, $j\geq i+\delta$ since $V_j$ is an independent set of $G[E_j]$. 
This helps us show the sharper bound of $|\ehat_\ehh| = O(n)$ in the following.

Without loss of generality assume that $h$ is a multiple of $\delta$ (if not, add at most $\delta-1$
empty layers $\ehat_{h+1}, \ehat_{h+2},\ldots$ to ensure that this is the case).
Let 
$$\beta(j) = \bigcup_{i=(j-1)\delta + 1}^{j\delta} \ehat_i\qquad\mbox{ for }j= 1, 2, \ldots, \frac{h}{\delta}$$ 
be a partition of the layers $\ehat_i$ into \textit{bands} of $\delta$ consecutive layers.
Let $\ehat_{odd} = \cup_{j: odd} \beta(j)$ and $\ehat_{even} = \cup_{j: even} \beta(j)$.

\begin{lemma}\label{lemma:layerlinear}
  $|\ehat_{odd}| = O(n)$, $|\ehat_{even}| = O(n)$ and therefore $|\ehat| = O(n)$.
\end{lemma}
\begin{proof}
  We prove the claim for $\ehat_{odd}$. The proof is essentially the same for $\ehat_{even}$.
  We aim to prove the following claim by induction on $k$ (for odd $k$): for some constant $C > 0$,
  \begin{equation}
    \label{eqn:indStep}
    \left|\bigcup_{j:odd  \ge k} \beta(j)\right| \le C \cdot \left|\bigcup_{j:odd \ge k} V(j)\right|,
  \end{equation}
  where $V(j)$ is the set of vertices such that every vertex in $V(j)$ has some 
  incident edge in $\beta(j)$. 
  Setting $k = 1$ in the above inequality, we see that $|\ehat_{odd}| = |\cup_{j:odd  \ge k} \beta(j)| = O(n)$.
  To prove the base case, let $k'$ be the largest odd integer less than or equal to $h/\delta$. Then, 
  $\cup_{j :odd\ge k'} \beta(j) = \beta(k')$  and 
  $\cup_{j :odd\ge k'} V(j) = V(k')$.
  Consider a vertex $v \in V(k')$.
  By Lemma \ref{lemma:neighborhood}, there are at most $c_3$ edges incident on $v$ from any layer.
  There are $\delta$ layers in $\beta(k')$ and therefore there are at most $c_3 \delta$ edges
  from $\beta(k')$ incident on any vertex $v \in V(k')$.
  Hence, $|\beta(k')| \le c_3 \delta |V(k')|$.
  Therefore, for any constant $C \ge c_3\delta$, it is the case that
  $|\cup_{j \ge k'} \beta(j)| \le C \cdot |\cup_{j \ge k'} V(j)|$.

  Taking (\ref{eqn:indStep}) to be the inductive hypothesis, let us now consider 
  $|\cup_{j \ge k-2} \beta(j)|$.
  Then,
  \begin{equation}
    \left|\bigcup_{j:odd  \ge k-2} \beta(j)\right|  \le  \left|\bigcup_{j:odd \ge k} \beta(j)\right| + |\beta(k-2)|
    \le  C \cdot \left|\bigcup_{j:odd \ge k} V(j)\right| + c_3\delta\cdot |V(k-2)|.
  \end{equation}
  The second inequality is obtained by applying the inductive hypothesis and the inequality $|\beta(k-2)|
  \le c_3 \delta |V(k-2)|$.
  By Lemma \ref{lemma:neighborhood}, at most half the vertices in $V(k-2)$ appear in $\cup_{j \ge k} V(k)$.
  Therefore, $|V(k-2) \setminus (\cup_{j \ge k} V(j))| \ge |V(k-2)|/2$.
  Hence,
  $$\left|\bigcup_{j:odd \ge k-2} \beta(j)\right|  \le C\cdot \left|\bigcup_{j:odd \ge k} V(j)\right| + 2c_3\delta \cdot \left|V(k-2) \setminus (\bigcup_{j:odd \ge k} V(j))\right|.$$
  Picking $C \ge 2c_3\delta$, we then see that
  $$\left|\bigcup_{j:odd \ge k-2} \beta(j)\right|  \le C \cdot \left(\left|\bigcup_{j:odd \ge k} V(j)\right| + \left|V(k-2) \setminus \left(\bigcup_{j:odd \ge k} V(j)\right)\right|\right) = C\cdot \left|\bigcup_{j:odd \ge k-2} V(j)\right|.$$
  The result follows by induction.
\end{proof}

\subsection{Many MIS Computations in Parallel}

In this section, we argue that 
Algorithm~\ref{algo:lowDimensionalMIS} \textsc{LowDimensionMIS} can be executed
on the $O(\log n)$ different distance threshold graphs in parallel on a congested clique. 
Table~\ref{tab:messages} shows number of messages sent/received per node 
in the execution of Algorithm~\ref{algo:lowDimensionalMIS} and
from this it is easy to see that Line 8 of Phase 2 can be executed as it is using  \lra~in $O(1)$ rounds for all the $O(\log n)$ layers in parallel due to their low communication 
requirements.  
For Lines 4-6 of Phase~2 we do the following load balancing via a 
\textit{designated receiver scheme}: each vertex has to send at most 
$O(n^{1/4}\log n)$ messages in an execution of Phase~2 for a layer. 
Therefore, for $O(\log n)$ layers one node is responsible of sending $O(n^{1/4}\log^2 n)$ messages. 
There are only $\lceil 2 \log n \rceil$ receivers needed for in an execution at a layer. 
For all layers the number of receivers needed are $O(\log^2 n)$. 
Hence we can designate different receivers such that no receiver gets more than $O(n)$ messages in execution of Phase~2 for all layers. 
Similar designated receiver scheme is applied for the execution of Phase 1. 

For parallel execution of Line 9 (\textsc{SequentialMIS}) of Phase 4 for all $O(\log n)$ layers 
we use the following \textit{message encoding scheme}: 
Each vertex $v$ constructs a $O(\log n)$-length bit string specifying 1 at position $\ell$ if $v$ is in MIS for the layer $\ell$ otherwise 0. 
Each vertex $v$ broadcasts this string. 
For a layer $\ell$, each vertex considers only $\ell^{th}$ bit of this 
message. 

\begin{table}[thb]
  \caption{Number of messages sent/received per node in the execution of
    Algorithm~\ref{algo:lowDimensionalMIS}\label{tab:messages}}
    \centering
    \begin{tabular}{|l|l|p{0.13\textwidth}|l|l|l|} \hline
      Phase & Line & Analysis & \parbox[t]{0.21\textwidth}{Number of messages to 
      send per node} &
      \parbox[t]{0.11\textwidth}{Number of receivers} &
      \parbox[t]{0.22\textwidth}{Number of messages to receive per receiver} \\ \hline
      1 & 2-4 & Lemma~\ref{lemma:si} & $O(n^{1/2})$ & 
      $n^{1/2}$ & $O(n)$ \\ \hline
      \multirow{2}{*}{2} & 4-6 & Lemma~\ref{lemma:wimis} & $O(n^{1/4}\log n)$ & $\lceil 2\log n \rceil $ & $O(n)$ \\
      & 8 & Lemma~\ref{lemma:qrule} & $O\left(\poly(\log n)\right)$ &  $n$ & $O\left(\poly(\log n)\right)$ \\ \hline
      3 & - & Thm.~\ref{thm:logstar} & $O(n^{1/2}\poly(\log^*n))$ & $ n $ & $O(n^{1/2}\poly(\log^*n))$ \\ \hline
\multirow{2}{*}{4} & 3 & Lemma~\ref{lemma:constantRounds} & $O(1)$ & 1 & $O(n)$ 
\\ 
 & 9 & Lemma~\ref{lemma:constantRounds} & 1 (1-bit) & $n$ & $n$ \\ \hline 
    \end{tabular}
  \end{table}

\section{Constant-Approximation to MFL}
\label{sec:facilitylocation}

Berns et al.~\cite{berns2012arxiv,berns2012facloc} showed how to compute a constant-factor approximation
to MFL in expected $O(\log\log n)$ rounds. (The algorithm presented in \cite{berns2012facloc}
runs in expected $O(\log\log n \cdot \log^* n)$ rounds, but this was subsequently
improved to expected $O(\log\log n)$ in \cite{berns2012arxiv}.)
A high level description of this algorithm is as follows.
Each node $v$ locally computes a value $r_v \ge 0$ that is a function of its opening cost $f_v$ and
distances to other nodes $\{d(v, w) \mid w \in V\}$.
Nodes with similar $r_v$-values join the same class; more precisely, a node $v$ with $3^k \cdot r_m \le r_v \le 3^{k+1} \cdot r_m$,
joins a class $V_k$.
Here $r_m$ is the minimum $r_u$-value over all nodes $u \in V$.
For nodes in each class $V_k$, we construct a graph $H_k = (V_k, E_k)$, where the edge-set $E_k$ is defined as
$\{\{u, v\} \mid u, v \in V_k, d(u, v) \le r_u + r_v\}$.
In the rest of the algorithm, in order to figure out which nodes to open as facilities,
the algorithm computes a $t$-ruling set on each graph $G_k$.
Analysis in \cite{berns2012arxiv,berns2012facloc} then shows that the solution to facility location produced by
this algorithm is an $O(t)$-approximation.
In \cite{berns2012arxiv} it is shown how to compute a 2-ruling set in expected $O(\log\log n)$ rounds on a
congested clique.
Since the classes $V_k$ form a partition of the nodes, the ruling set computations occur on disjoint sets
of nodes and can proceed in parallel.
This leads to a constant-factor approximation to MFL in expected $O(\log\log n)$ rounds.

The 3-ruling set algorithm and the MIS algorithm in the present paper can replace the slower 2-ruling set 
and this yields the following result.

\begin{theorem}
  There exists a distributed algorithm that computes a constant-approximation to
  the metric facility location problem (w.h.p.) in the
  congested-clique model and which has an expected running time of
  $O(\log \log \log n)$ rounds.
  Additionally, if the input metric space has constant doubling dimension then 
  a constant-approximation can be computed in constant rounds (w.h.p.)
  \label{theorem:FacLocApprox}
\end{theorem}

\section{Conclusion}

In a recent paper, Drucker et al.~\cite{DruckerKuhnOshmanPODC2014} show that the congested clique can simulate powerful classes of bounded-depth circuits, implying that even slightly super-constant lower bounds for the congested clique would give new lower bounds in circuit complexity. This provides some explanation for why there are no
non-trivial lower bounds in the congested clique model. One could view this result as providing motivation
for proving even stronger upper bounds. As shown in this paper, it is possible to design algorithms
that run significantly faster than $\Theta(\log\log n)$ rounds for well-known problems.
Continuing this program, we are interested in designing algorithms running in $o(\log\log n)$ rounds for
MST and related problems such as connectivity verification.

\subsubsection*{Acknowledgments.}
We would like to thank reviewers of DISC 2014 for their careful reading and thoughtful comments. 

\bibliography{disc}

\end{document}